\documentclass[a4paper,runningheads]{llncs}
\usepackage{fullpage}

\usepackage{amsmath}
\usepackage{amssymb}
\usepackage{graphicx}
\usepackage{xspace}
\usepackage{enumerate}
\usepackage{comment}
\usepackage{cite}
\usepackage{algorithm}
\usepackage{subfig}
\usepackage[noend]{algpseudocode}
\usepackage[usenames]{xcolor}
\usepackage{stmaryrd} 
\usepackage{mathtools} 

\spnewtheorem{prop}{Property}{\bfseries}{\itshape} 

\newcommand{\remove}[1]{}

\renewenvironment{proof}
{{\em Proof.\ }}{\hspace*{\fill}$\Box$\par\vspace{2mm}}

\renewcommand{\epsilon}{\varepsilon}

\begin{document}
\title{A Lower Bound on the Diameter of the Flip Graph
}
\author{Fabrizio Frati\inst{1}}
\institute{Dipartimento di Ingegneria, University Roma Tre, Italy\\
\email{frati@dia.uniroma3.it}}
\maketitle

\begin{abstract}
The flip graph is the graph whose nodes correspond to non-isomorphic combinatorial triangulations and whose edges connect pairs of triangulations that can be obtained one from the other by flipping a single edge. In this note we show that the diameter of the flip graph is at least $\frac{7n}{3} + \Theta(1)$, improving upon the previous $2n + \Theta(1)$ lower bound.
\end{abstract}

\section{Introduction} \label{se:introduction}

A {\em combinatorial triangulation} is a maximal planar graph (a planar graph to which no edge can be added without destroying planarity) together with a clockwise ordering for the edges incident to each vertex. An intuitive way to define a combinatorial triangulation is as an equivalence class of planar drawings (say on the sphere) of a maximal planar graph, where two drawings are equivalent if a continuous morph exists from one drawing to the other that does not create crossings or overlaps between edges. We are interested in {\em simple} combinatorial triangulations, which have no self-loops or multiple edges. In the following, when we say {\em triangulation} we always mean {\em simple combinatorial triangulation}. Observe that, in a planar drawing equivalent to a triangulation, all the faces are delimited by cycles with three vertices (hence the name {\em triangulation}).

Consider a planar drawing $\Gamma$ on the sphere equivalent to a triangulation $G$ and consider an edge $(a,b)$ in $G$. If $(a,b)$ were removed from $\Gamma$, there would exist a unique region of the sphere delimited by a cycle with four edges; in fact the cycle delimiting such region would be $(a,a',b,b')$, for some vertices $a'$ and $b'$. The operation of {\em flipping} $(a,b)$ consists of removing $(a,b)$ from $G$ and inserting the edge $(a',b')$ inside the region delimited by the cycle $(a,a',b,b')$. The resulting triangulation $G'$ might not be simple though. In the following, we only refer to flips that maintain the triangulations simple.

The {\em flip graph} ${\cal G}_n$ describes the possibility of transforming $n$-vertex triangulations using flips. The vertex set of ${\cal G}_n$ is the set of distinct $n$-vertex triangulations; two triangulations $G$ and $H$ are connected by an edge in ${\cal G}_n$ if there exists an edge $e$ of $G$ such that flipping $e$ in $G$ results in $H$.

Various properties of the flip graph have been studied. A particular attention has been devoted to the {\em diameter} of ${\cal G}_n$, which is the length of the longest (among all pairs of vertices) shortest path; refer to the surveys~\cite{bh-fpg-09,bv-hfct-11}. A first proof that the diameter of ${\cal G}_n$ is finite goes back to almost a century ago~\cite{w-jdmv-36}. A sequence of deep improvements~\cite{bjrsv-mt4uf-14,chktw-adfdht-15,k-dfts-97,mno-dfhts-03,nn-dtgcs} have led to the current best upper bound of $5n+\Theta(1)$, which was proved this year by Cardinal et al.~\cite{chktw-adfdht-15}. Significantly less results and techniques have been presented for the lower bound. We are only aware of a $2n+\Theta(1)$ lower bound on the diameter of ${\cal G}_n$, which was proved by Komuro~\cite{k-dfts-97} by exploiting the existence of triangulations with ``very different'' vertex degrees. The main contribution of this note is the following theorem.

\begin{theorem} \label{th:lower-bound}
For every $n\geq 3$, the diameter of the flip graph is at least $\frac{7n}{3}-34$.
\end{theorem}

\section{Proof of the Main Result} \label{se:lower-bound}

In this section we prove Theorem~\ref{th:lower-bound}. Let $n\geq 3$. For a triangulation $G$, we denote by $V(G)$ and $E(G)$ its vertex and edge set, respectively.

Consider any $n$-vertex triangulation $G_1$. A path incident to $G_1$ in ${\cal G}_n$ is a sequence of $n$-vertex triangulations such that the first triangulation in the sequence is $G_1$ and any two triangulations which are consecutive in the sequence can be obtained one from the other by flipping a single edge. Thus, a path incident to $G_1$ in ${\cal G}_n$ corresponds to a {\em valid sequence} $\sigma=(u_1,v_1),\dots,(u_{k},v_{k})$ of flips, where $u_1,\dots,u_{k},v_1,\dots,v_{k}$ are vertices in $V(G_1)$ and $(u_i,v_i)$ is an edge of the triangulation obtained starting from $G_1$ by performing flips $(u_1,v_1),\dots,(u_{i-1},v_{i-1})$ in this order. For a valid sequence $\sigma$ of flips, denote by $G_1^{\sigma}$ the $n$-vertex triangulation obtained starting from $G_1$ by performing the flips in $\sigma$. Observe that $V(G_1)=V(G_1^{\sigma})$, given that a flip only modifies the edge set of a triangulation, and not its vertex set.

Now consider any two $n$-vertex triangulations $G_1$ and $G_2$ and consider a simple path in ${\cal G}_n$ between them. This path corresponds to a valid sequence $\sigma$ of flips transforming $G_1$ into $G_2$. By the definition of ${\cal G}_n$, the $n$-vertex triangulations $G_1^{\sigma}$ and $G_2$ are isomorphic; that is, there exists a bijective mapping $\gamma:V(G_1^{\sigma}) \rightarrow V(G_2)$ such that $(u,v)\in E(G_1^{\sigma})$ if and only if $(\gamma(u),\gamma(v))\in E(G_2)$.

The key idea for the proof of Theorem~\ref{th:lower-bound} is to consider the bijective mapping $\gamma$ {\em before} the flips in $\sigma$ are applied to $G_1$ and to derive a lower bound on the number of flips in $\sigma$ based on properties of $\gamma$. In fact, the property we employ is the number of common edges of $G_1$ and $G_2$ according to~$\gamma$.

More precisely, for a bijective mapping $\gamma:V(G_1) \rightarrow V(G_2)$ between the vertex sets of two triangulations $G_1$ and $G_2$, we define the {\em number $c_{\gamma}$  of common edges with respect to $\gamma$} as the number of distinct edges $(u,v)\in E(G_1)$ such that $(\gamma(u),\gamma(v))\in E(G_2)$. We have the following.

\begin{lemma} \label{le:lower-bound}
For any two $n$-vertex triangulations $G_1$ and $G_2$, the number of flips needed to transform $G_1$ into $G_2$ is at least $3n-6-\max_{\gamma}{c_{\gamma}}$, where the maximum is over all bijective mappings $\gamma:V(G_1) \rightarrow V(G_2)$.
\end{lemma}

\begin{proof}
The statement descends from the following two observations. First, two isomorphic $n$-vertex  triangulations have $3n-6$ common edges according to the bijective mapping $\gamma$ realizing the isomorphism. Second, for any two $n$-vertex triangulations $H$ and $L$ that have $\ell$ common edges with respect to any bijective mapping $\gamma$, flipping any edge in $H$ results in a combinatorial triangulation $H'$ such that $H'$ and $L$  have at most $\ell+1$ common edges with respect to $\gamma$.
\end{proof}

It remains to define two $n$-vertex  triangulations $G_1$ and $G_2$ such that {\em any} bijective mapping $\gamma$ between their vertex sets has a small number $c_{\gamma}$ of common edges.

\begin{itemize}
\item Triangulation $G_1$ is defined as follows (see Fig.~\ref{fig:g1}). Let $H$ be any triangulation of maximum degree six with $\lfloor \frac{n}{3} \rfloor +2$ vertices. Note that the number of faces of $H$ is $2(\lfloor \frac{n}{3} \rfloor +2)-4=2\lfloor \frac{n}{3} \rfloor$. If $n\equiv 2$ modulo $3$, if $n\equiv 1$ modulo $3$, or if $n\equiv 0$ modulo $3$, then insert a vertex inside each face of $H$, insert a vertex inside each face of $H$ except for one face, or insert a vertex inside each face of $H$ except for two faces, respectively. When a vertex is inserted inside a face of $H$, it is connected to the three vertices of $H$ incident to the face. Denote by $G_1$ the resulting $n$-vertex triangulation. We say that the vertices of $G_1$ in $H$ are {\em blue}, while the other vertices of $G_1$ are {\em red}.
\item Triangulation $G_2$ is defined as follows (see Fig.~\ref{fig:g2}). Starting from a path $P$ with $n-2$ vertices, connect all the vertices of $P$ to two further vertices $a$ and $b$, and connect $a$ with $b$.
\end{itemize}

\begin{figure}[b]
  \centering
  \subfloat[]{\label{fig:g1}\includegraphics[width=0.4\textwidth]{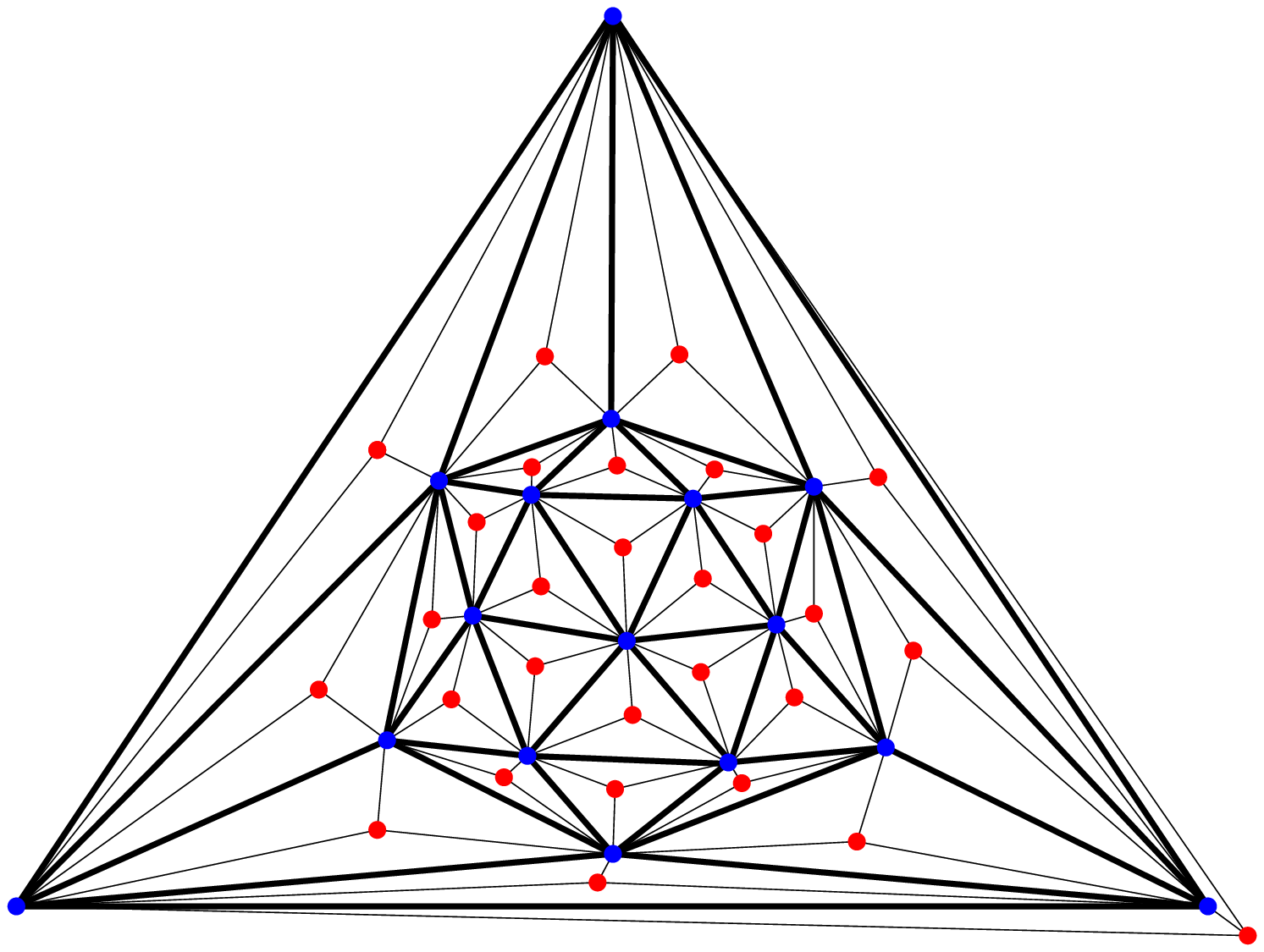}}\hfil%
  \subfloat[]{\label{fig:g2}\includegraphics[width=0.4\textwidth]{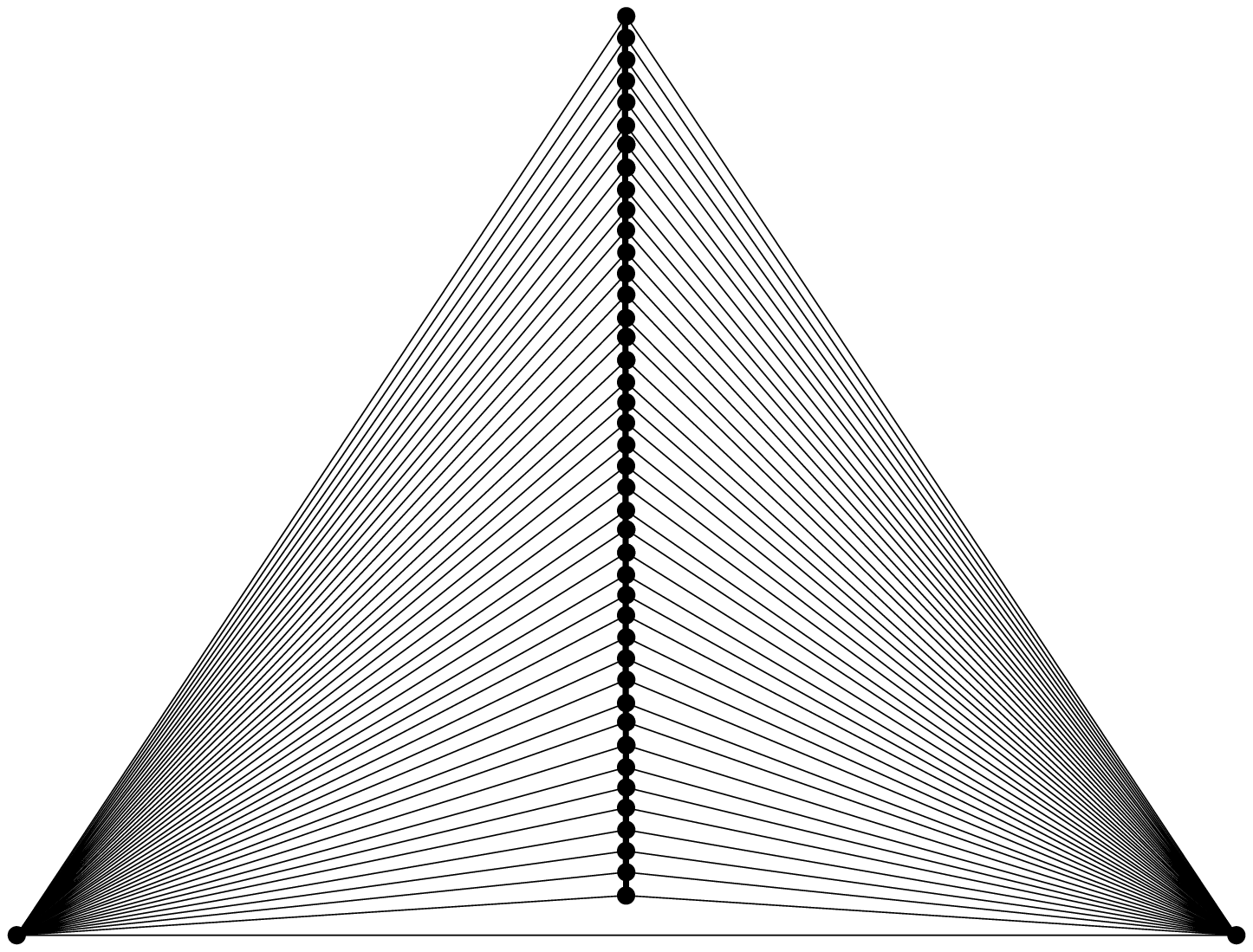}}\hfil%
  \caption{Triangulations $G_1$ (a) and $G_2$ (b).}
  \label{fig:graphs}
\end{figure}

We have the following.

\begin{lemma} \label{le:g1vsg2}
For any bijective mapping $\gamma:V(G_1) \rightarrow V(G_2)$, we have $c_{\gamma}\leq 2\lfloor \frac{n}{3} \rfloor+28$.
\end{lemma}

\begin{proof}
Consider any bijective mapping $\gamma:V(G_1) \rightarrow V(G_2)$. First, note that each vertex $v\in V(G_1)$ has degree at most twelve. Namely, $v$ has at most six blue neighbors; further, $v$ has at most six incident faces in $H$, hence it has at most six red neighbors. It follows that, whichever vertex in $V(G_1)$ is mapped to $a$ according to $\gamma$, at most twelve out of the $n-1$ edges incident to $a$ are shared by $G_1$ and $G_2$ with respect to $\gamma$. Analogously, at most twelve out of the $n-1$ edges incident to $b$ are shared by $G_1$ and $G_2$ with respect to $\gamma$. It remains to bound the number of edges of $P$ that are shared by $G_1$ and $G_2$ with respect to $\gamma$. This proof uses a pretty standard technique (see, e.g.,~\cite{b-spcp-61,chktw-adfdht-15}). Since $G_1$ has no edge connecting two red vertices, the number of edges of $P$ that are shared by $G_1$ and $G_2$ with respect to $\gamma$ is at most the number of edges of $P$ that have at least one of their end-vertices mapped to a blue vertex; since $\lfloor \frac{n}{3} \rfloor +2$ vertices of $G_1$ are blue, there are at most $2\lfloor \frac{n}{3} \rfloor +4$ such edges of $P$. It follows that the number of edges shared by $G_1$ and $G_2$ with respect to $\gamma$ is at most $2\lfloor \frac{n}{3} \rfloor+28$.
\end{proof}

By Lemma~\ref{le:g1vsg2}, we have that $G_1$ and $G_2$ are two $n$-vertex triangulations such that, for any bijective mapping $\gamma:V(G_1) \rightarrow V(G_2)$, we have $c_{\gamma}\leq  2\lfloor \frac{n}{3} \rfloor+28$. By Lemma~\ref{le:lower-bound}, the number of flips needed to transform $G_1$ into $G_2$ is at least $3n-6-2\lfloor \frac{n}{3}\rfloor-28\geq \frac{7n}{3}-34$. This concludes the proof of Theorem~\ref{th:lower-bound}.

\section{Conclusions} \label{se:conclusions}

In this note we have presented a lower bound of $\frac{7n}{3} +\Theta(1)$ on the diameter of the flip graph for $n$-vertex triangulations. One of the main ingredients for this lower bound is a lemma stating that there exist two $n$-vertex triangulations such that any bijective mapping $\gamma$ between their vertex sets creates at most $c_{\gamma}\leq \frac{2n}{3} +\Theta(1)$ common edges.

It not clear to us whether the bound resulting from this approach can be improved further. That is, is it true that, for every two $n$-vertex triangulations, there exists a bijective mapping $\gamma$ between their vertex sets creating $c_{\gamma}\geq \frac{2n}{3} +\Theta(1)$ common edges? The only lower bound on the value of $c_{\gamma}$ we are aware of comes as a corollary of the fact that every $n$-vertex triangulation has a matching of size at least $\frac{n+4}{3}$ as proved in~\cite{bddfk-tbmmm-04}, hence $c_{\gamma}\geq \frac{n+4}{3}$.

It is an interesting fact that, for every $n$-vertex triangulation $H$, a bijective mapping $\gamma: V(H)\rightarrow V(G_2)$ exists creating $c_{\gamma}= \frac{2n}{3} +\Theta(1)$ common edges, where $G_2$ is the graph from the proof of Theorem~\ref{th:lower-bound}. In fact, every $n$-vertex triangulation $H$ has a set of $\frac{n}{3} +\Theta(1)$ vertex-disjoint simple paths covering its vertex set $V(H)$, as proved by Barnette~\cite{b-tpg-65} (this bound is the smallest possible~\cite{b-spcp-61}). Mapping these paths to sub-paths of the path $P$ in $G_2$ provides the desired bijective mapping $\gamma$.

\subsubsection*{Acknowledgments} Thanks to Michael Hoffmann for an inspiring seminar he gave when visiting the author.

\bibliographystyle{abbrv}
\bibliography{bibliography}

\end{document}